\documentclass[conference,a4paper, final]{IEEEtran}
\usepackage{enumitem}

\usepackage[footnote,marginclue,nomargin,final]{fixme}

\usepackage{comment}
\usepackage{graphics}
\usepackage{float}				
\usepackage{tikz}				
\usepackage{amssymb}
\usepackage{bm}					
\usepackage{sistyle}			
\SIproductsign{\cdot}
\usepackage{bbm}
\usepackage[cmex10]{amsmath}
\usepackage{algorithmic}
\usepackage[tight,footnotesize]{subfigure}
\usepackage{stfloats}
\usepackage[utf8]{inputenc}	

\usepackage{algorithm}
\usepackage{amsthm}

\newtheorem{theorem}{Theorem}
\newtheorem{proposition}{Proposition}
\theoremstyle{remark}

\providecommand{\indi}[1]{\mathbbm{1}\left\{#1\right\}}

\newcommand{\E}[1]{\mathbb{E}\left[#1\right]}
\newcommand{\EE}[2]{\mathbb{E}_{#1}\left[#2\right]}

\newcommand{\pr}[1]{\text{Pr}\left[#1\right]}
\newcommand{\mbf}[1]{\mathbf{#1}}

\newcommand{\unit}[1]{\ 
	\ifmmode & \left[\textup{#1}\right] \hspace{1cm}
	\else \hfill $\left[\textup{#1}\right]$ \hspace{1cm}
	\fi
}

\makeatletter
 \begingroup
 \catcode`\_=\active
 \protected\gdef_{\@ifnextchar|\subtextup\sb}
 \endgroup
 \def\subtextup|#1|{\sb{\textup{#1}}}
 \AtBeginDocument{\catcode`\_=12 \mathcode`\_=32768 }
\makeatother
\begin{document}

\sloppy

\title{Block-Fading Channels with Delayed CSIT at Finite Blocklength}

\author{
  \IEEEauthorblockN{Kasper F. Trillingsgaard and
    Petar Popovski}
  \IEEEauthorblockA{Department of Electronic Systems, Aalborg University,
    Aalborg, Denmark}

}



\maketitle

\begin{abstract} In many wireless systems, the channel state information at the
  transmitter (CSIT) can not be learned until \textit{after} a transmission has
  taken place and is thereby outdated. In this paper, we study the benefits of
  delayed CSIT on a block-fading channel at finite blocklength. First, the
  achievable rates of a family of codes that allows the number of codewords to
  expand during transmission, based on delayed CSIT, are characterized. A
  fixed-length and a variable-length characterization of the rates are provided
  using the dependency testing bound and the variable-length setting introduced
  by Polyanskiy et al. Next, a communication protocol based on codes with
  expandable message space is put forth, and numerically, it is shown that
  higher rates are achievable compared to coding strategies that do not benefit
  from delayed CSIT.
\end{abstract}

\section{Introduction}
The success of wireless high-speed networks is largely based on
reliable transmission of \emph{large} data packets through the use of the
principles from coding and information theory. On the other hand, many emerging
applications that involve machine-to-machine (M2M) communication rely on
transmission of very short data packets with strict deadlines, where the
asymptotic information-theoretic results are not applicable. The fundamentals of
such a communication regime have recently been addressed  in~\cite{finite_rate}, where it was shown that the rates achievable by
fixed-length block codes in traditional point-to-point communication can be
tightly approximated by
\begin{align}
  R^*(n,\epsilon) &= C - \sqrt{\frac{V}{n}}Q^{-1}(\epsilon) +
  \mathcal{O}\left(\frac{\log n}{n}\right),\label{eq:normal_approximation}
\end{align}
where $C$ is the Shannon capacity, $V$ is the channel dispersion, $n$ is the
blocklength, $\epsilon$ is the desired probability of error and $Q^{-1}(\cdot)$
the inverse of the standard Q-function.
In \cite{feedback} it was shown that allowing the use of variable-length
stop-feedback (VLSF) coding improves the achievable rates dramatically, seen through the fact 
that the dispersion term in
\eqref{eq:normal_approximation} vanishes.

Finite blocklength analysis is particular interesting for fading
channels. Whereas the effect of fading may often be averaged when
blocklengths tend to infinity, fading may have severe impact on the achievable
rates when blocklengths are small, i.e. as the blocklength decreases and/or the
coherence time of the block-fading channel increases, the worst-case channel
conditions largely dictate the achievable rates
\cite{quasi_static,scalar_dispersion,block_fading}. In such cases it is
beneficial to use variable-length coding and allow a transmission
that experiences good channel realization to terminate early. 
However, sending an ACK/NACK  at an arbitrary instant is rather impractical. From a system design
perspective, it is viable to assume that a feedback opportunity occurs regularly
after each $T-$th channel use, through which the sender gets either ACK or NACK,
along with the delayed CSIT about the transmission conditions in the block. This
deteriorates the benefits of VLSF, as
the sender may continue to send incremental redundancy to the receiver until the
next feedback opportunity occurs. The problem is circumvented by the concept of
\textit{backtrack retransmission} (BRQ) \cite{backtrack}, described as
follows. Upon receiving delayed CSIT and NACK, the sender estimates how much side
information is required by the receiver to decode the packet, 
erroneously received in the previous block. If the side information is
less than the total number of source bits that can be sent in the next block,
new source bits are appended to the side information and then jointly
channel-coded. Thus, the sender \emph{expands} the original message by appending
new source bits, before the original message has been decoded.

This paper generalizes the concept of BRQ to the case of finite blocklength. We
consider a block-fading channel with two states, where the receiver has full
channel state information (CSI) and the transmitter learns the CSI after
transmission in each block. We introduce a family of codes, termed \emph{expandable message
 space (EMS) codes}, that allows the message space to expand upon reception of
a delayed CSIT. The EMS codes allow the transmitter to expand the
number of codewords in a tree-like manner. Using these codes, we propose a
communication scheme, based on BRQ, which improves the achievable rates
by expanding the message space according to the delayed CSIT.

We illustrate the concept of backtrack retransmission and EMS codes for
block-fading through the following example.  Consider a block-fading channel in
which a transmission is allowed to take at most two blocks of $T$ channel
uses. Using fixed-length block codes, the transmitter may either send a packet
of $b_1$ source bits in one block with a target probability of error $\epsilon$,
or it may transmit a packet in both blocks of $b_2>2b_1$ source bits, i.e. a
higher rate, with the same target error probability $\epsilon$ at the cost of
twice the blocklength. A na\"{i}ve variable-length coding can be applied as
follows: the transmitter sends a packet in the first block of $b_{1v}>b_1$
source bits and obtains the delayed CSIT of the first block after
transmission. If the CSIT allows decoding with a probability of error less than
$\epsilon$, the packet is decoded and otherwise incremental redundancy is
transmitted in the second block, leading to half the rate,
$\frac{b_{1v}}{2T}$. $b_{1v}$ is chosen to match the target error probability
$\epsilon$. In this paper we aim beyond this na\"{i}ve scheme and investigate
how to use the second (and the subsequent) blocks in the best possible way given
that a transmission has already taken place in the first block. Using the EMS
codes, which are introduced in this paper, the transmitter may send a packet of
$b_{1e}$ source bits in the first block, and if the CSI does not allow reliable
decoding, the transmitter may combine incremental redundancy and a new message
of $b_{2e}$ source bits in the second block. $b_{1e}$ and $b_{2e}$ can thereby
be jointly optimized to obtain the highest rate under the constraint that the
probability of error is smaller than $\epsilon$.


In contrast to traditional variable-length codes with stop-feedback as analyzed
in \cite{feedback}, neither the amount of source information to be transmitted
nor the blocklength is known in advance for our communication scheme. Clearly,
this implies some practical issues in the higher communication layers but also
provides improved achievable rates.

\begin{figure}[!t]
  \centering
  \includegraphics[width=2.5in]{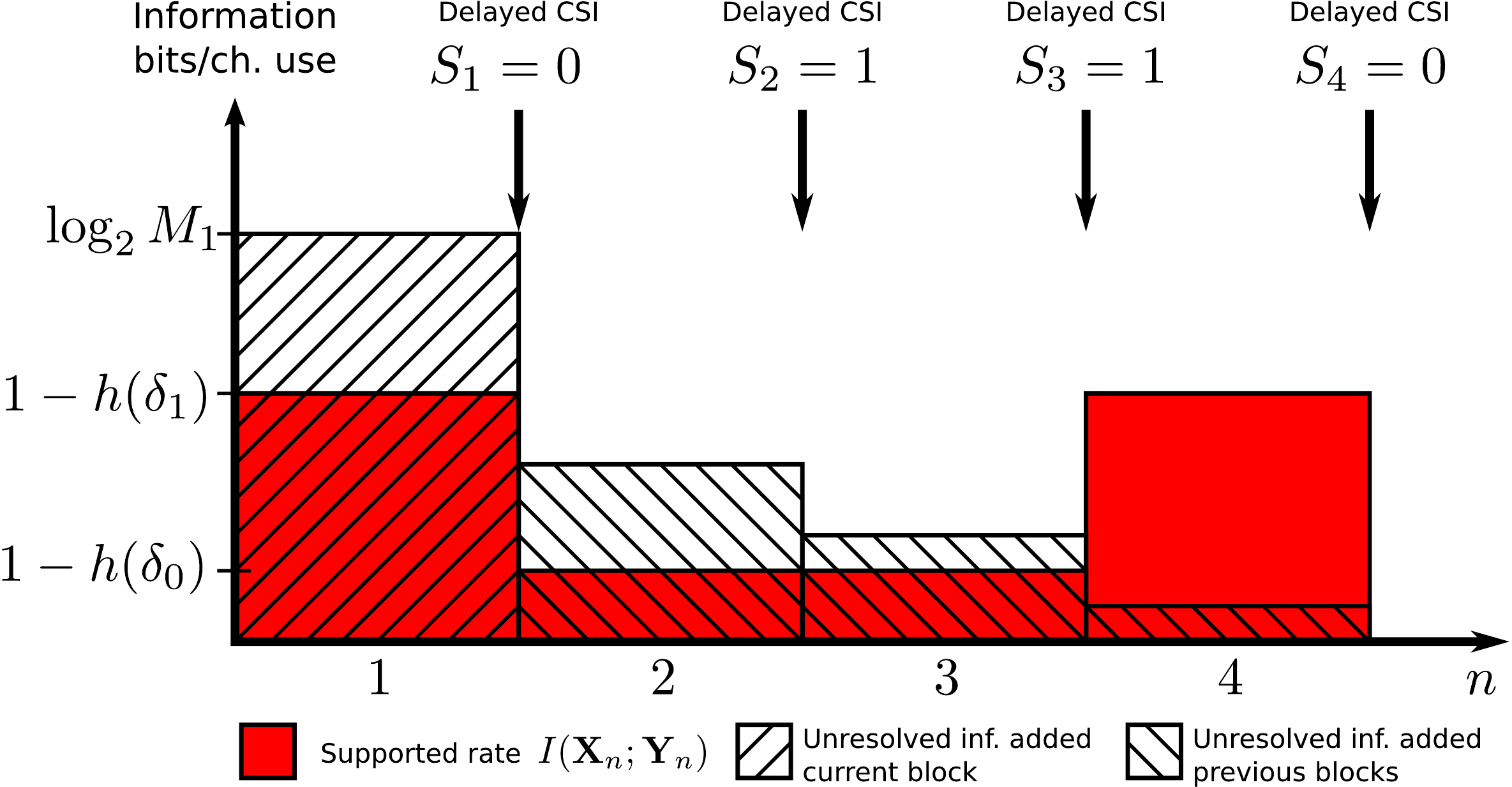}
  \caption{Variable-length transmission scheme based on delayed CSIT.}
  \label{fig:brq_comp}
\end{figure}

\section{The Block-Fading model}
\label{sec:channel_model}
We consider a single-user binary channel with block-fading. The channel has two
states in which it acts as binary symmetric channels (BSC) with different
crossover probabilities. Transmissions to the receiver are done in blocks, each
consisting of $T$ channel uses, where $T$ models the coherence time of the
block-fading channel. The blocks are enumerated $n\in \mathbb{N}$, where
$\mathbb{N}$ is the set of natural numbers, and the channel input in the $n$-th
block is denoted $\mathbf{X}_n\in\mathcal{X}=\{0,1\}^T$. The binary channel
state in the $n$-th block, denoted by $S_n\in\{0,1\}$, is a random variable that
is independent of any previous channel states and distributed according to
$\pr{S_n=1}=1-\pr{S_n=0}=q$, $q\in[0,1]$. In state $s\in\{0,1\}$, the channel
acts as a BSC with crossover probability $\delta_s\in[0,\frac{1}{2}]$ such that
receiver obtains $\mbf{Y}_n = \mbf{X}_n \oplus \mbf{Z}_n\in
\mathcal{Y}=\{0,1\}^T$ where $\oplus$ denotes the XOR operation and
$\mbf{Z}_n\in\{0,1\}^T$ is a binary noise vector with iid entries distributed as
$\text{Bern}(\delta_s)$. For the remaining part of this paper, we assume,
without loss of generality, that $\delta_1<\delta_0$. Note that the state with
$\delta_0$ can also be used to model a block with intermittent interference. The
receiver knows $S_n$ when $\mbf{X}_n$ is sent, but the transmitter learns it after it
is sent, at the end of the $n$-th block. Thus the channel state $S_n$ may be
used to adapt the transmission scheme from the $(n+1)$-th block. Moreover, we
consider a \emph{stop-feedback setting} in which the receiver feeds back
noiseless ACK/NACK to the transmitter that indicates termination.  The capacity
of the described channel is $C = 1 - q h_b(\delta_1) - (1-q) h_b(\delta_0)$,
where $h_b(\cdot)$ is  the binary entropy function. The capacity is 
approached by fixed-length codes, with blocklengths tending to infinity. As shown in
Section~\ref{sec:results}, binary fading markedly degrades the performance at
finite blocklength.


\subsection{Fixed-length block codes}
The achievable rates using fixed-length block codes may be approximated by
\eqref{eq:normal_approximation}, where the dispersion of the block-fading model
described (per channel use) can be shown to be
\cite{scalar_dispersion}
\begin{align}
  V&=\E{ \delta_S (1-\delta_S) \left(\log
      \frac{1-\delta_S}{\delta_S}\right)^2}\nonumber\\
  &\quad+ T q(1-q)(h_b(\delta_0)-h_b(\delta_1))^2,\label{eq:dispersion3}
\end{align}
with the expectation taken over the channel state $S\in\{0,1\}$.

\subsection{Variable-length codes}
For fading channels with CSI at the receiver
(CSIR) and delayed CSIT, variable-length coding can be achieved either through
stop-feedback as in \cite[Theorem 3]{feedback} or delayed CSIT. With
stop-feedback, the transmitter sends incremental redundancy until an ACK
is obtained, while the receiver makes an estimate of the correct codeword in
each block and if the reliability of the estimate is higher than $1-\epsilon$,
the receiver feeds back an ACK that terminates the transmission. This
scheme is referred to as variable-length stop-feedback (VLSF) coding. On the other hand, delayed CSIT can be used by the transmitter to estimate
whether the receiver has collected enough information density and thereby
terminate the transmission. This communication scheme is referred to as
variable-length coding with delayed CSIT (VLD).

The achievable rates of these communication schemes can be computed using the
Theorem~3 in \cite{feedback} and the dependency testing bound in
\cite{finite_rate}, respectively. In Fig.~\ref{fig:brq_comp}, the VLD scheme is
illustrated for the block-fading model described previously. Initially, the
transmitter chooses a codeword from a codebook of $M_1$ codewords. In each
block, the receiver collects information density that resolves some uncertainty
about the correct codeword. After the transmission, the transmitter
obtains the CSI of the previous block, and using the dependency testing bound in
\cite{finite_rate}, computes the probability of error $\hat{\epsilon}$. If $\hat{\epsilon}<\epsilon$, the transmitter terminates the
transmission, while it sends additional incremental redundancy
otherwise. This continues until the amount of information density allows the receiver to reliably decode the message. As
shown in the specific realization in Fig.~\ref{fig:brq_comp}, variable-length
coding with periodic feedback eventually leads to cases where the receiver
collects a wasteful amount of information density.

\section{Expandable Message Space Codes}\label{sec:codes_exp}
\begin{figure}[!t]
\centering
\subfigure[Tree of codewords.]{
  \includegraphics[width=1.5in]{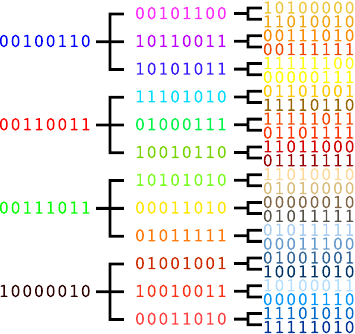}
	\label{fig:codetree}
}
\subfigure[List of codewords.]{
  \includegraphics[width=1.46in]{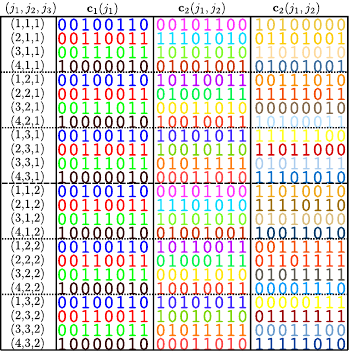}
	\label{fig:codelist}
}
\caption{Depicts a random binary tree codebook and the corresponding list of
  codewords where $\mathcal{X}=\{0,1\}^8$  and message cardinalities $M_1=4$, $M_2=3$ and
  $M_3=2$.  Note that equal blocks have the same colors. }
\label{fig:code_design}
\vspace{-12pt}
\end{figure}
This section describes the EMS codes and EMS stop-feedback (EMS-SF) codes. In
contrast to fixed-length block codes, EMS codes allow the number of codewords in
each block to expand in a tree-like fashion. Two types of codes are introduced
which allow the message space to expand in each block and are analogous to
fixed-length block codes and VLSF codes in \cite{feedback}, respectively. In the
following, $M_i^k$ denotes $\prod_{n=i}^k M_n$ for $k\geq i$ and $1$ otherwise.
An $( (M_1,\ldots,M_N),  \epsilon)$ EMS code consists of 
\begin{itemize}
\item $N$ message sets,
$\mathcal{M}_n=\{1,\ldots, M_n\}$, $n\in\{1,\ldots,N\}$,
\item a set of encoding functions $\mathbf{c}_n:
\mathcal{M}_1\times\ldots \times \mathcal{M}_n \rightarrow \mathcal{X}$,
\item a decoder function $\mbf{g}: \mathcal{Y}^N \rightarrow \mathcal{M}_1\times\ldots
\times \mathcal{M}_N \cup \{\mathrm{e}\} $ that assigns estimates $\hat j_1,
\ldots,\hat j_N$ or an error message $\mathrm{e}$ to each received sequence
$\mathbf{y}^N$,
\end{itemize}
s.t. the average error probability $\pr{ \mbf{g}(\mbf{Y}^N)\not=
  (J_1,\ldots,J_N) } \leq \epsilon$, where
$J_1\in\mathcal{M}_1,\ldots,J_N\in\mathcal{M}_N$ denote the transmitted
equiprobable messages.

In particularly, we denote a codeword of an EMS code as
$\mbf{c}(j_1,\ldots,j_N)\in\mathcal{X}^N$, with $(j_1,\ldots,
j_N)\in(\mathcal{M}_1 ,\ldots, \mathcal{M}_N)$, and
$\mbf{c}_i^k(j_1,\ldots,j_N)\in\mathcal{X}^{k-i+1}$, $k\geq i$, denote the
$i$-th to the $k$-th block of $\mbf{c}(j_1,\ldots,j_N)$. As opposed to a
traditional fixed-length block code with $\prod_{i=1}^N M_i$ messages, the EMS
codes differ only by the definition of the encoder functions that restricts the
$n$-th block to only depend on the messages $M_1,\ldots,M_n$. This property
implies that the codewords of an EMS code have a tree-like overlapping structure
such that
\begin{align}
  \mbf{c}_1^n(j_1,\ldots,j_n,j_{n+1},\ldots,j_N)=\mbf{c}_1^n(j_1,\ldots,j_n),
\end{align}
for $(j_{n+1},\ldots,j_N)\in(\mathcal{M}_{n+1},\ldots,\mathcal{M}_N)$, and hence we can uniquely denote $\mbf{c}_1^n(j_1,\ldots,j_N)$ by
$\mbf{c}(j_1,\ldots,j_n)$.

Next, we define an EMS code that takes
advantage of stop-feedback. An $(l, (M_1,\ldots,M_N), \epsilon)$ EMS-SF code
consists of
\begin{itemize}
\item $N$ message sets $\mathcal{M}_n\in\{1,\ldots,M_n\}$, $n\in\{1,\ldots,N\}$,
\item a sequence of encoders $\mbf{c}^{SF}_n: \mathcal{M}_1 \times \ldots \times
  \mathcal{M}_{n'}\rightarrow \mathcal{X}$, where $n'=\min(n, N)$,
\item a sequence of decoders $\mbf{g}^{SF}_n: \mathcal{Y}^n \rightarrow
  \mathcal{M}_1\times \ldots \times \mathcal{M}_{n'}$, that assigns the best
  estimates $\hat j_1,\ldots,\hat j_{n'}$ at time $n$ for each possible sequence
  in $\mathcal{Y}^n$,
\item a random variable $\tau^*\in\mathbb{N}$ satisfying $\E{\tau^*} \leq l$,
\end{itemize}
s.t. $\pr{\mathbf{g}^{SF}_{\tau^*}(Y^{\tau})\not = (J_1, \ldots, J_{{\tau^*}'})}\leq
\epsilon$, ${\tau^*}'=\min(\tau^*, N)$, where
$J_1\in\mathcal{M}_1,\ldots,J_N\in\mathcal{M}_N$ denotes the equiprobable messages. 

Although $N$ message sets are defined for the EMS-SF codes, the transmission may
be terminated before all messages have been decoded without declaring an
error. When $M_n=1$ for $n\geq 2$, the EMS code and EMS-SF are identical to
traditional a fixed-length block code and a VLSF code with $M_1$ messages,
respectively.  The overlapping property allows EMS and EMS-SF codes to be built
online, based on common randomness, according to feedback or CSI after each
block. A practical EMS example are the rate-compatible convolutional codes in
which new source bits only affect the future states.

To illustrate how EMS codes can be used in variable-length coding on a binary
block-fading channel with $T=8$, consider the example on
Fig.~\ref{fig:code_design}. Assume that from a codebook of $M_1=4$ codewords,
the transmitter initially transmits a codeword with index
$j_1\in\{1,2,3,4\}$. After the transmission, the transmitter obtains delayed CSIT
and finds that the codeword can not be decoded reliably at the receiver. Instead
of sending incremental redundancy, the transmitter chooses to expand the
message space by a factor of $M_2=3$ and injects a new message
$j_2\in\{1,2,3\}$. Fig.~\ref{fig:codetree} depicts how the codebook expands. To
send the second block, the encoder $\mbf{c}_2(j_1,j_2)$ is used, and hence the
second transmitted block depends on both $j_1$ and $j_2$. The dependency on $j_1$
essentially combines the injected message $j_2$ and the incremental redundancy
for the first message $j_1$. Note that if $M_2$ had been $1$, purely incremental
redundancy would have been send. Upon obtaining the CSI of the second block, the
transmitter injects another message $j_3\in\{1,2\}$ using the encoder
$\mbf{c}_{3}(j_1,j_2,j_3)$, and finally the CSI allows reliable decoding. The
codebook of $M_1 M_2 M_3=24$ codewords generated through this process is shown
in Fig~\ref{fig:codelist}. 

For the remaining results, we use random codebooks with iid entries drawn from a
$\text{Bern}(\frac{1}{2})$ distribution.

In order to provide a non-asymptotic characterization of the achievable rates of
the EMS codes, we state the following bound, analogous to the dependency testing bound in \cite{finite_rate}.
\begin{theorem}\label{thm:joint_brq_decode} The error probability for
  $((M_1,\ldots,M_N), \epsilon)$ EMS code is bounded as
\begin{align}
  \epsilon &\leq \pr{ i(\mbf{X}_1^N ; \mbf{Y}_1^N) \leq \log \frac{M_1^N-1}{2} }\nonumber\\
  &\quad+ \sum_{n=1}^{N}  \frac{ M_{n+1}^N (M_{n}-1)}{2}
  \nonumber\\
  &\qquad\qquad\pr{
    i(\mbf{X}_1^N ; \mbf{Y}_1^{n-1} ,\overline{\mbf{Y}}_{n}^N) >
    \log \frac{M_1^N-1}{2} },\label{eq:thm1_error_bound}
\end{align}  
where $\mbf{X}_i$ is distributed according to the channel input distribution
 and $\mbf{Y}_i$ and
$\overline{ \mbf{Y}}_i$ are distributed according to the output distribution,
conditioned and unconditioned on the channel input, respectively.
\end{theorem}
\begin{proof}
  See Appendix~A.
\end{proof}
The result in Theorem~\ref{thm:joint_brq_decode} can be equivalently stated as
\begin{proposition}
  The error probability for $((M_1,\ldots,M_N), \epsilon)$-code EMS code is
  \begin{align}
    \epsilon &\leq \sum_{n=1}^{N} \frac{M_{n+1}^N (M_{n}-1)}{M_1^N-1}\nonumber\\
    & \qquad\E{\exp{\left\{ i(\mbf{X}_1^{n-1} ; \mbf{Y}_1^{n-1}) - \left| i(\mbf{X}_1^N ;
          \mbf{Y}_1^N) - \log \frac{M_1^N-1}{2}
        \right|^+\right\}}}.\label{eq:prop1_statement}
  \end{align}
  \label{thm:joint_brq_decode2}
\end{proposition}
\begin{proof}
  See Appendix~B.
\end{proof}

Next, we consider a non-asymptotic bound for EMS-SF codes. This
generalizes the VLSF code from \cite{finite_rate}.
\begin{theorem}
  Fix $\gamma_n$ for $n\in\{1,\ldots,N \}$ and set $\gamma_n=0$ for $n>N$. Let $\mbf{X}_i$ and
  $\overline{ \mbf{X}}_i$ be independent copies of the same process and $\mbf{Y}_i$ be the output
  of the channel when $\mbf{X}_i$ is its input. Define the hitting times
  \begin{align}
    \tau & = \inf\left\{ k \geq 1 : i(\mbf{X}_1^k;\mbf{Y}_1^k) \geq \sum_{i=1}^{k}\gamma_i\right\}\label{eq:tau}\\
    \overline \tau_n & = \inf\left\{ k \geq n :   i(\mbf{X}_1^{n-1} ; \mbf{Y}_1^{n-1})+i(\mbf{\overline{X}}_n^{N} ; \mbf{Y}_n^{N}) \geq \sum_{i=1}^{k}\gamma_i\right\},\label{eq:overlinetau}
  \end{align}
  Then for any tuple $(M_1,\ldots,M_{N})$ there exists an $(l, (M_1,\ldots,M_N),
  \epsilon)$ EMS-SF code such that
  \begin{align}
   l \leq \E{\tau} \quad \text{and} \quad \epsilon \leq
    \E{\sum_{n=1}^{\min(\tau,N)} \tilde M^{\tau}_n\pr{\overline \tau_n
        \leq \tau}}.\label{eq:vlf_eps}
  \end{align}
  where $\tilde M_{k}^{n} =   (M_k - 1)M_{k+1}^{\min(N,n))}$.
  \label{thm:brq_vlf}
\end{theorem}
\begin{proof}
  See Appendix~C.
\end{proof}
The achievable rate is then computed as $R^*(\epsilon) = \frac{\E{\sum_{n=1}^{\min(\tau,N)}\log M_n}}{\E{\tau}}$.
To enable efficient computation of the bound in Theorem~\ref{thm:brq_vlf}, we
loosen the bound on the error probability \eqref{eq:vlf_eps} as follows
\begin{align}
  \epsilon &\leq \E{\sum_{l=1}^{\tau} \tilde M^{\tau}_l\pr{\overline \tau_l \leq
      \tau}}\\
&= \sum_{n=1}^{\infty}\pr{\tau = n} \sum_{l=1}^{n} \tilde M_l^n \E{\indi{\tau
      \leq
      n}\exp\{ -i(X_l^{n} ; Y_l^{n})\}}\\
  &\leq \sum_{n=1}^{\infty}\pr{\tau = n} \sum_{l=1}^{n} \tilde M_l^n
  \exp\left\{-\sum_{i=l}^n \gamma_i\right\}\label{eq:loosened_epsilon}
\end{align}
\section{Backtrack retransmission}
\label{sec:backtrack}

\begin{algorithm}[!t]
  \caption{BRQ with delayed CSIT at the transmitter.}
  \label{alg:brq_trans}
\begin{algorithmic}[1]
  \STATE \textbf{input}: $M_1$ and target probability of
  error $\epsilon$.
  \STATE \textbf{initialize}: Generate random binary codebook with $M_1$
  codewords $\{\mbf{c}(j_1)\}$. Fetch $\log M_1$ nats, corresponding to the
  message $j_1\in\mathcal{M}_1$, and transmit
  $\mbf{x}=\mbf{c}(j_1)\in\{0,1\}^T$.
  \STATE Receive the channel state $S_{1}$. Terminate if $\epsilon_{S_1}(\{M_1\})<\epsilon$.
  \STATE $k\gets 1$
  \REPEAT
  \STATE $k\gets k+1$
  \STATE Compute
  $\tilde \epsilon_{k}=\epsilon_{S_1,\ldots,S_{k-1},1}(\{M_1,\ldots,M_{k-1},1\})$. 
  \IF[Expand message space]{$\tilde \epsilon_{k}\leq \epsilon$} 
  \STATE Find $M_{k}$ such that $\epsilon_{S_1,\ldots,S_{k-1},1}(\{M_1,\ldots,M_{k}\})=\epsilon$.
  \ELSE[Pure incremental redundancy]
  \STATE $M_{k}\gets 1$
  \ENDIF
  \STATE Expand codebook by a factor of $M_{k}$ by generating $M_{k}$ codewords
  $\{\mathbf{c}_k(j_1,\ldots,j_{k})\}_{j_{k}\in\mathcal{M}_{k}}$
 for each $j_1\in\mathcal{M}_1,\ldots,j_{k-1}\in\mathcal{M}_{k-1}$.
  \STATE Fetch $\log M_{k}$ nats and set the message $j_{k}$
  accordingly. Transmit $\mbf{x}_{k}=\mathbf{c}_k(j_1,\ldots,j_{k})$.
  \STATE Receive the channel state $S_{k}$.
    \UNTIL{$S_k=1$ and $\tilde \epsilon_k \leq \epsilon$}
  \end{algorithmic}
\end{algorithm}

The key idea of BRQ is to reduce the collected amount of wasteful information
density by increasing the number of codewords in the codebook during
transmission. At short blocklengths, this can efficiently be achieved using EMS
codes.  As for VLSF and VLD, we propose two different communication schemes
which are based on delayed CSIT alone and a combination of delayed CSIT and
stop-feedback.
\subsection{BRQ with Delayed CSIT}
The operation of BRQ
is illustrated in Fig.~\ref{fig:brq} and the protocol at the transmitter is
summarized in Algorithm~\ref{alg:brq_trans}. We assume that the transmitter and
receiver have exchanged a seed to generate
common randomness (for codebooks), and $\epsilon_{s_1,\ldots,s_{k}}(\{M_1,\ldots,M_{k}\})$
denotes the achievable error probability, computed by
Theorem~\ref{thm:joint_brq_decode}, of an EMS code with messages
$M_1,\ldots,M_k$ and the state sequence $s_1,\ldots,s_k$ on the block-fading
channel. The transmitter initiates the transmission in block $k=1$ by choosing a
codeword from a random codebook of $M_1$ codewords.
By the end of the $k$-th block, the CSI $s_k$ is
obtained at the transmitter. Based on the CSI $s_1,\ldots,s_k$, the objective of
the transmitter is to ensure that the decoder will not collect wasteful
information density in block $k+1$. Therefore the transmitter computes the
probability of error if the receiver were to decode by the end of block $k+1$ and
the CSI turns out to be $S_{k+1}=1$. This probability of error is given by
$\tilde \epsilon_{k+1}=\epsilon_{s_1,\ldots,s_k,1}(M_1, \ldots,M_k,1)$. If
$\tilde \epsilon_{k+1}<\epsilon$, higher reliability than necessary is achieved
if $S_{k+1}=1$, and the message space is thus expanded by a factor of
$M_{k+1}$. $M_{k+1}$ is computed such that, if $S_{k+1}=1$, then the messages
$M_1,\ldots,M_{k+1}$ can be jointly decoded with a probability of error
$\epsilon$, i.e. $\epsilon_{S_1,\ldots,S_k,1}(M_1,
\ldots,M_k,M_{k+1})=\epsilon$. Otherwise, $M_{k+1}$ is set to $1$, and purely
incremental redundancy is send.  Termination occurs when
$\epsilon_{s_1,\ldots,s_k}(M_1, \ldots,M_k)\leq \epsilon$.  Since the
transmitter only uses delayed CSIT, the transmitter and receiver may generate
the same codebooks using common randomness.
\begin{figure}[!t]
  \centering
  \includegraphics[width=2.3in]{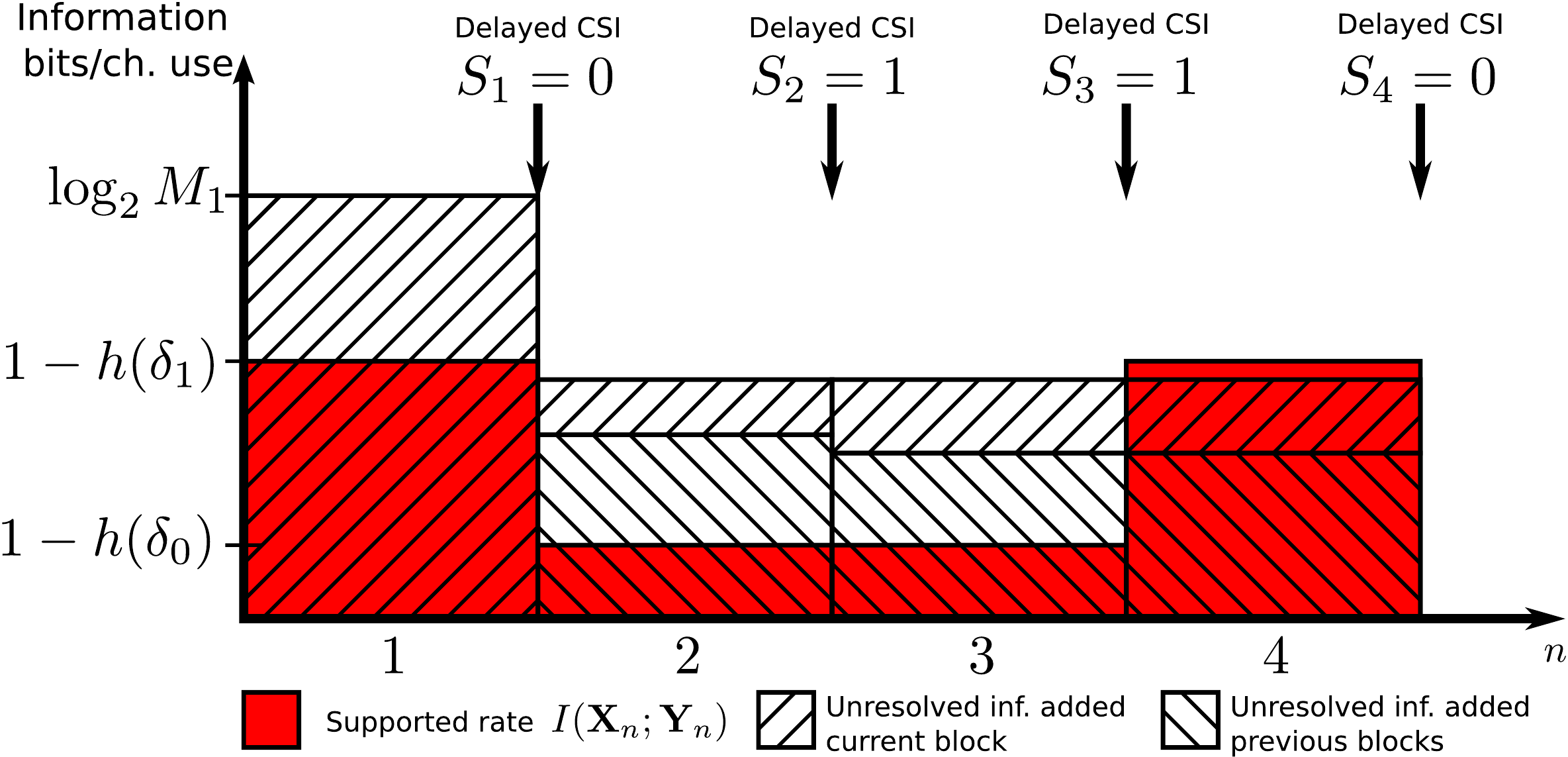}
\caption{Operation of BRQ.}
\label{fig:brq}
\end{figure}

\subsection{BRQ with Delayed CSIT and Stop-Feedback}
When both stop-feedback and delayed CSIT is available at the transmitter, the
proposed BRQ scheme is similar to Algorithm~\ref{alg:brq_trans} but uses an
EMS-SF code. However, for stop-feedback codes, the receiver decides whether to
decode based on its received signal $\mbf{Y}^n$ or, for the codes constructed in
Theorem~\ref{thm:brq_vlf}, when the information density surpasses a
threshold. Wasteful information density is thereby reflected by an overwhelming
probability of decoding in a specific block. With $\tau^{\{M_1,\ldots,M_k\}}$
being the random stopping time of an $((M_1,\ldots,M_k), \epsilon)$ EMS-SF code,
the probability of decoding at the end of block $k$, given NACKs
were received in the first $k-1$ blocks, is denoted by
\begin{align}
& p^{\{M_1,\ldots,M_k\}}_{s_1,\ldots,s_k}=\pr{\tau^{\{M_1,\ldots,M_k\}} = k \big|
  \tau^{\{M_1,\ldots,M_k\}}\geq k\right.\nonumber\\
&  \qquad\qquad\qquad\qquad\qquad \left., (S_1,\ldots,S_k)=(s_1,\ldots,s_k)}
\end{align}
where $s_1,\ldots,s_k$ denotes the state sequence of the block-fading
channel. We introduce an additional parameter $\beta$ which serves as a
threshold for when to expand the message space.

Therefore, the transmitter uses the following algorithm; in slot $k$, if
$p^{\{M_1,\ldots,M_{k-1},1\}}_{s_1,\ldots,s_{k-1},1}>\beta$, the message
space is expanded by a factor of $M_k$ s.t.
$p^{\{M_1,\ldots,M_k\}}_{s_1,\ldots,s_{k-1},1}=\beta$. Otherwise, $M_k$ is set
to $1$. After computation of $M_k$, the threshold value $\gamma_k$ in
Theorem~\ref{thm:brq_vlf} may be computed using \eqref{eq:loosened_epsilon} such
that the receiver chooses to decode when the probability of error is less than
$\epsilon$. Using this transmission protocol, the probability of error never
exceeds $\epsilon$ and the probability of decoding in a specific slot does not
exceed $\beta$.

\section{Numerical Results}\label{sec:results}
\begin{figure}[!t]
\centering
\subfigure[$T=100$]{
  \includegraphics[width=2.7in]{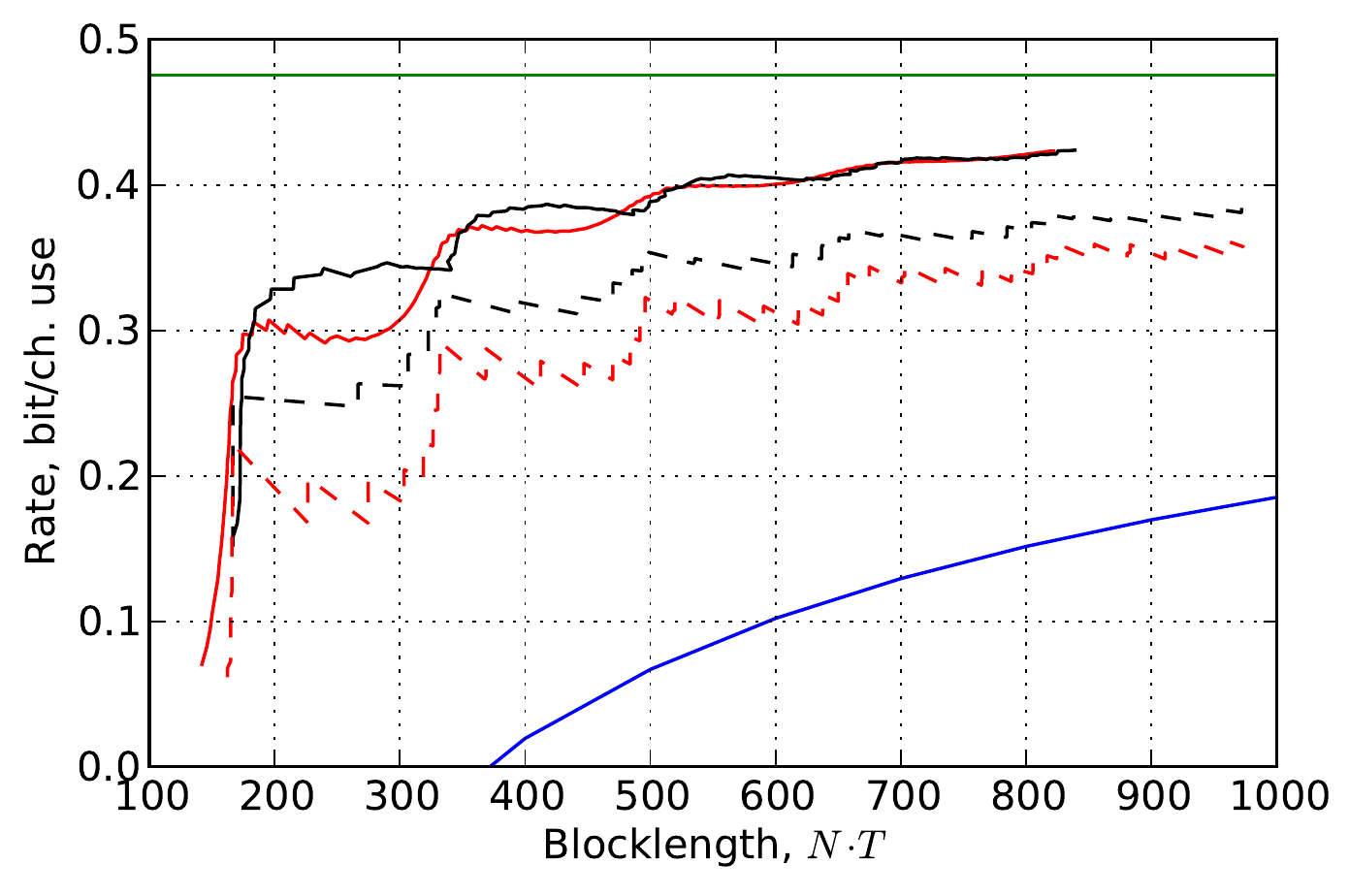}
	\label{fig:results1}
}
\subfigure[$T=200$]{
  \includegraphics[width=2.7in]{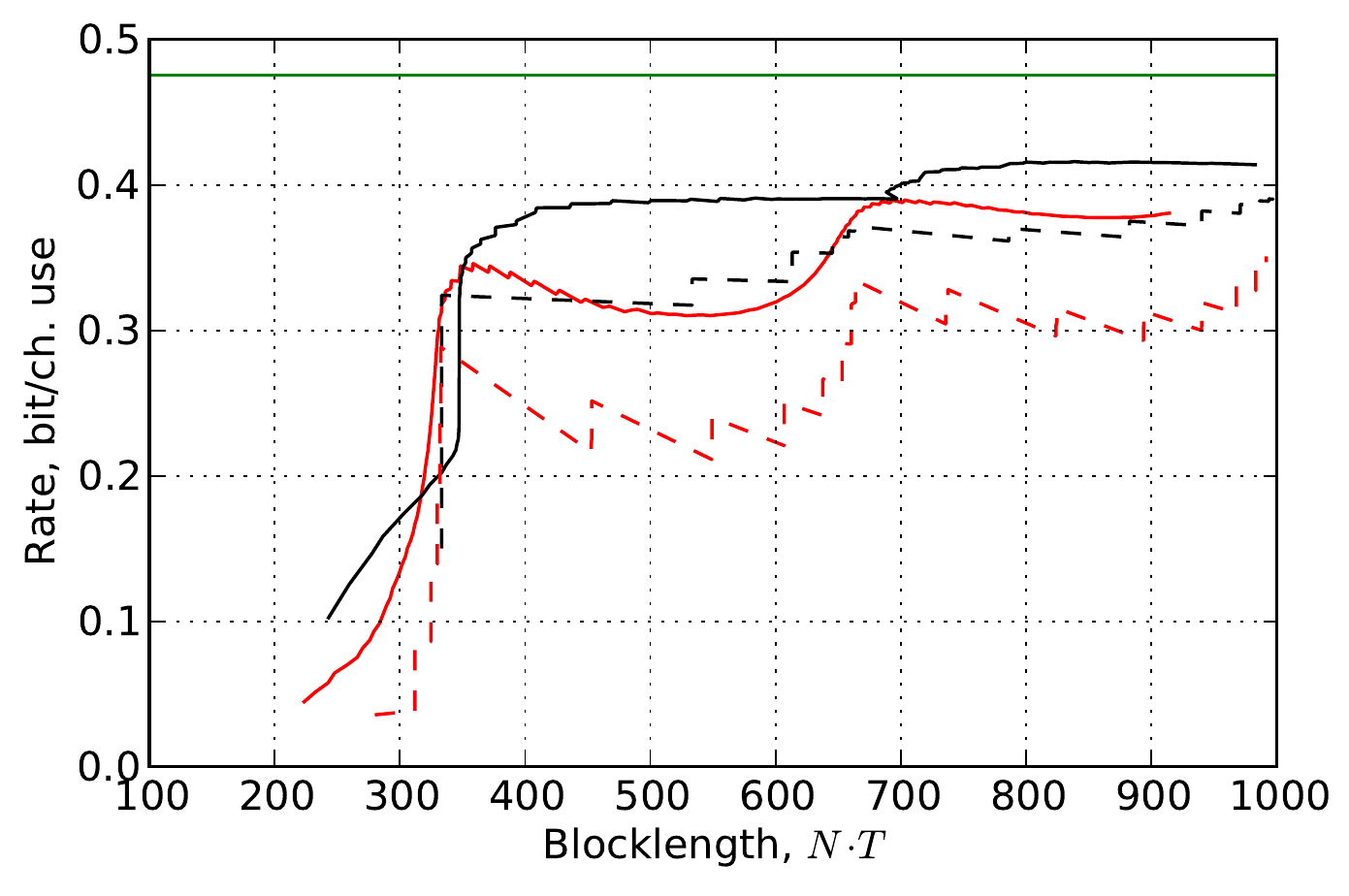}
	\label{fig:results2}
}
\caption{Achievable rates for the block-fading model with the parameters
  $\delta_1=0.05, \delta_0=0.30, q=0.6$ and $\beta=0.9$. Green: capacity, black (solid):
BRQ with delayed CSIT and stop-feedback, black (dashed): BRQ with delayed CSIT,
red (solid): VLSF scheme, red (dashed): VLD scheme, blue: finite-length block
codes (not in plot for $T=200$). }
\label{fig:result}
\end{figure}
To assess the performance of the proposed communication schemes, the achievable
rates of fixed-length block codes, the VLSF scheme, the BRQ schemes are
computed.

The achievable rates of fixed-length block codes are computed using the normal
approximation in \eqref{eq:normal_approximation}. For the remaining schemes, the
achievable rates are computed using Theorem~\ref{thm:joint_brq_decode} and
Theorem~\ref{thm:brq_vlf} and by averaging over all fading realizations for a range of
$M_1$ values. To reduce the computational complexity of averaging, we restrict
the number of message space expansions to $5$. The achievable rates are computed
using the parameters $\delta_1=0.05$, $\delta_0=0.3$, $q=0.6$ and
$\beta=0.9$. Computed achievable rates are shown in Fig.~\ref{fig:result}. 

Observe that schemes based on variable-length coding in general outperforms the
fixed-length block codes. Moreover, for the VLSF and VLD schemes, we see that
periodic decoding implies that the achievable rates have decreases in some
ranges of blocklengths which becomes more pronounced with higher coherence time
$T$. In these ranges of blocklengths, the BRQ schemes
achieve higher rates. Note that optimization over $\beta$ may yield better rates
for BRQ with delayed CSIT and stop-feedback.

\section{Discussion and Conclusions}
\label{sec:conclusions}
In this paper, we considered binary block-fading fading channel with two
states. A family of codes, EMS codes, that allows the message space to expand
during transmission was introduced and we provided bounds on the probability of
error. Using these codes, we proposed two transmission schemes based the
backtrack retransmission scheme. Numerical results showed that the proposed
communication schemes achieve better rates, for the specific parameters, than
communication schemes that do not benefit from delayed CSIT.





\bibliography{isit}

\bibliographystyle{IEEEtran}

\appendices
\section{Proof of Theorem 1}
\begin{proof}
  The proof is based on the proof of the DT-bound in \cite{finite_rate}.

  \textit{Codebook generation}: Generate the random codebook according to the
  following algorithm:
  \begin{enumerate}
  \item Let $n=1$. Generate $M_1$ codewords $\mbf{c}(j_1)\in\mathcal{X}$ for
    $j_1\in\mathcal{M}_1$ according to the distribution $P_{\mbf{X}_1}$.
  \item For each tuple $(j_1,\ldots,j_n)\in(\mathcal{M}_1,\ldots,\mathcal{M}_n)$,
    generate $M_{n+1}$ codewords
    $\mbf{c}_{n+1}(j_1,\ldots,j_{n+1})\in\mathcal{X}$ according to the
    distribution $P_{\mbf{X}_{n+1}}$.
  \item Let $n=n+1$. If $n<N$, goto step 2, otherwise stop.
  \end{enumerate}

  \textit{Encoder}: The transmitter maps the message
  $(j_1, \ldots,j_N)\in(\mathcal{M}_1,\ldots,\mathcal{M}_L)$ to the codeword
  $\mbf{c}(j_1,\ldots,j_N)$ which is transmitted.

\textit{Decoder}: The decoder uses the Feinstein suboptimal decoder
\cite{finite_rate}. Let $\{Z_{\mbf{x}^N}\}$, $\mbf{x}^N\in\mathcal{X}^{N}$, be a
collection of functions defined as
\begin{align}
  Z_{\mbf{x}^N}(\mbf{y}^N) = \indi{i(\mbf{x}^N ; \mbf{y}^N) > \log \frac{M_1^N-1}{2}}.
\end{align}
These functions are likelihood ratio hypothesis tests.  The decoder runs through
all $\prod_{n=1}^N M_n$ codewords and performs likelihood ratio hypothesis
tests, and the codeword corresponding to the lowest index such that
$Z_{\mbf{x}^N}(\mbf{y}^N)=1$ is output.

An ordering of the codewords $\{\mbf{c}(j_1,\ldots,j_N)\}$ is defined such that
$\mbf{c}_j=\mbf{c}(j_1,\ldots,j_N)$ if and only if 
\begin{align}
j_1+\sum_{n=2}^N M_1^{n-1} (j_n-1)=j,
\end{align}
for $j\in\{1,\ldots,M_1^N\}$. This ordering corresponds to the codebook shown in
Fig.~\ref{fig:codelist}.  The conditional probability of error given that the
$j$-th codeword was sent is
\begin{align}
  &\pr{ \{ Z_{\mbf{c}_j}(\mbf{Y}^N)=0 \} \cup\bigcup_{i<j} \{ Z_{\mbf{c}_i}(\mbf{Y}^N)=1 \} \big| \mbf{X}^N =
    \mbf{c}_j}\\
  &\quad = \pr{ i(\mbf{c}_j ; \mbf{Y}^N) \leq \log \frac{M_1^N-1}{2} \big| \mbf{X}^N =
    \mbf{c}_j}\nonumber\\
  &\qquad + \sum_{i<j}\pr{ i(\mbf{c}_i ;\mbf{Y}^N) > \log \frac{M_1^N-1}{2} \big| \mbf{X}^N =
    \mbf{c}_j}. \label{eq:prob_error1}
\end{align}
By symmetry in the codebook,  \eqref{eq:prob_error1} can be written as
\begin{align}
  &\leq  \pr{ i(\mbf{X}^N ; \mbf{Y}^N) \leq \log \frac{M_1^N-1}{2} }\nonumber\\
  &\quad+ \sum_{n=0}^{N-1} j'_{n} \pr{ i(\mbf{X}^N;\mbf{Y}_1^n, \overline{\mbf{Y}}_{n+1}^N) > \log \frac{M_1^N-1}{2} }\label{eq:prob_error2}
\end{align}
where
\begin{align}
  j'_{N-1}&=\left\lfloor \frac{j-1}{M_1^{N-1}} \right \rfloor \\
  j'_{n} &= \left\lfloor \frac{j-1}{M_1^n} \right
  \rfloor-\left\lfloor \frac{j-1}{M_1^{n+1}} \right
  \rfloor  \text{ for } n\in\{1,\ldots,N-2\}\\
  j'_0&=j-1 - \left\lfloor \frac{j-1}{M_1} \right
  \rfloor
\end{align}
Intuitively, $j'_l$ describes the number of codewords among the indices
$\{1,\ldots, j-1\}$ that shares the first $n$ blocks with the transmitted
codeword $\mbf{c}_j$.

Note that we have
\begin{align}
  \frac{1}{M_1^N} \sum_{j=1}^{M_1^N} \left\lfloor \frac{j-1}{M_1^n}
  \right \rfloor&=\frac{1}{M_{n+1}^N}\sum_{k=1}^{M_{n+1}^N}(k-1)\\
  &= \frac{M_{n+1}^N-1}{2}\label{eq:prf_11}
\end{align}
and
\begin{align}
\frac{M_{n+1}-1}{2} - \frac{
  M_{n+2}^N-1}{2}&=\frac{M_{n+2}^N \left(M_{n+1}-1  \right)}{2}.\label{eq:prf_12}
\end{align}
Thus by averaging over \eqref{eq:prob_error2} with respect to $j$ and by using
\eqref{eq:prf_11} and \eqref{eq:prf_12}, we obtain the upper bound in \eqref{eq:thm1_error_bound} on the probability
of error.
\end{proof}

\section{Proof of Proposition~1}
\begin{proof}
  Note that we can write \eqref{eq:thm1_error_bound} as
\begin{align}
&\epsilon\leq  \sum_{n=1}^{N} \frac{M_{n+1}^N (M_{n}-1)}{M_1^N-1}
\nonumber\\
&\qquad \E{\left( \pr{i(\mbf{X}_{1}^N;\mbf{Y}_{1}^N) <
      \log \frac{M_1^N-1}{2}}\right.\right.\nonumber\\ &
  \qquad\quad\left.\left.+\frac{M_1^N-1}{2}\pr{i(\mbf{X}_1^{n-1};\mbf{Y}_1^{n-1}) +
        i(\mbf{X}_{n}^N ;
        \overline{\mbf{Y}}_{n}^N)\right.\right.\right.\nonumber\\
  &\qquad\qquad\qquad\qquad\qquad\qquad\left.\left.\left.\geq \log \frac{M_1^N-1}{2}}
  \right)}.\label{eq:prop1_rewrite}
\end{align}
Using the identity
\begin{align}
  e^{\rho} \exp\left\{ - \left| \log \frac{z e^{\rho}}{\gamma }
    \right|^+\right\}= \mathbbm{1}\left(z \leq \frac{\gamma}{e^{\rho}}\right) + \frac{ \gamma}{z}
  \mathbbm{1}\left(z > \frac{\gamma}{e^{\rho}}\right),\label{eq:prop1_rewrite2}
\end{align}
by setting
$z=\frac{dP_{\mbf{X}_{n}\ldots,\mbf{X}_{N},\mbf{Y}_{n},\ldots,\mbf{Y}_{N}}}{d(P_{\mbf{X}_{n}\ldots,\mbf{X}_{N}}\times
  P_{\mbf{Y}_{n},\ldots,\mbf{Y}_{N}})}$ and by averaging both sides of
\eqref{eq:prop1_rewrite2} with respect to the joint pmf
$P_{\mbf{X}_{n}\ldots,\mbf{X}_{N},\mbf{Y}_n,\ldots,\mbf{Y}_{N}}$, we obtain
\begin{align}
  & \E{\exp\left\{\rho - \left| i(\mbf{X}_{n}^N;\mbf{Y}_{n}^N)+ \rho - \log \gamma
    \right|^+\right\}}\nonumber \\
  &\qquad= \pr{i(\mbf{X}_{n}^N; \mbf{Y}_{n}^N) + \rho \leq \log \gamma} \nonumber\\
  &\qquad\quad+ \gamma
  \pr{i(\mbf{X}_{n}^N ; \overline{\mbf{Y}}^N_{n}) + \rho > \log \gamma }.\label{eq:avg_identity}
\end{align}
Substituting \eqref{eq:avg_identity} into \eqref{eq:prop1_rewrite}
yields the bound in \eqref{eq:prop1_statement}.
\end{proof}

\section{Proof of Theorem 2}
\begin{proof}
  The proof is similar to the proof of Theorem~3 in \cite{feedback}.

  A codebook, shared by the transmitter and receiver, with $M_1^N$ infinite
  dimensional codewords drawn from the distribution $P_{\mathbf{X}_n}$  such that
  \begin{align}
    \mbf{c}_n(j_1,\ldots,j_N) &\in \mathcal{X}
  \end{align}
  for $(j_1,\ldots,j_{N})\in(\mathcal{M}_1, \ldots,
  \mathcal{M}_{N})$ and $n\in\mathbb{N}$. Additionally, the codebook has the following property
  \begin{align}
    \mbf{c}_n(j_1,\ldots,j_{n}) &= \mbf{c}_n(j_1,\ldots,j_{N})\in\mathcal{X}
  \end{align}
   and $n\in\{1,\ldots,N\}$. As for the EMS codes, this property implies that
   \begin{align}
     \mbf{c}_n(j_1,\ldots,j_k,j_{k+1},\ldots,j_{N})=\mbf{c}_n(j_1,\ldots,j_k,j'_{k+1},\ldots,j'_{N})
   \end{align}
   for $k\in\{1,\ldots,N-1\}$, $(j_{k+1}',\ldots,j_N')\in(\mathcal{M}_{k+1},\ldots,\mathcal{M}_N)$ and  $n\in\{1,\ldots,k\}$.

   The $(l, (M_1,\ldots,M_N),\epsilon)$ EMS-SF code is defined by a sequence of
   encoders $\mbf{c}^{SF}_n: \mathcal{M}_1\times \ldots, \mathcal{M}_{n'}
   \rightarrow \mathcal{X}$ that maps the messages $j_1, \ldots, j_{n'}$ to the
   channel input $\mbf{c}_n(j_1,\ldots,j_{n'})$, where  $n'=\min(n,N)$.  By the end of the $n$-th
   block, the decoder computes the $M_1^{n'}$ information densities
  \begin{align}
    S_n(j_1,\ldots,j_{n'}) 
    &= \sum_{k=1}^n i(\mbf{c}^{SF}_k(j_1,\ldots,j_{n'}) ; \mbf{Y}_k )
  \end{align}
  for $(j_1,\ldots,j_{n'})\in(\mathcal{M}_1,\ldots,\mathcal{M}_{n'})$. 

  Define the events $E_n(j_1,\ldots,j_{n'}) =\left\{ S_n(j_1,\ldots,j_{n'})\geq
    \sum_{i=1}^{n'}\gamma_i\right\}$ and the stopping times
  \begin{align}
    \tau_{j_1,\ldots,j_{n}} &= \inf\left\{ k \geq 1, \exists
      (j_{n+1},\ldots,j_{N})\in(\mathcal{M}_{n+1},\ldots,\mathcal{M}_N)\right.\nonumber\\
      &\qquad\qquad\qquad \left.: E_k(j_1,\ldots,j_n) \right\}
  \end{align}
for $n<N$ and otherwise   
\begin{align}
\tau_{j_1,\ldots,j_N} = \inf\left\{ k \geq 1: E_k(j_1,\ldots,j_N) \right\}.
  \end{align}
  We define the moment of the first upcrossing as
  \begin{align}
    \tau^* = \min \{n\geq 1: \exists j_1,\ldots,j_{n'}, E_n(j_1,\ldots,j_{n'}) \},
  \end{align}
  and the decoder is given by
  \begin{align}
    &\mbf{g}^{SF}_{\tau^*}(\mbf{Y}^{\tau^*}) \nonumber\\
    &=\max\{(j_1,\ldots,j_{{\tau^*}'})\in(\mathcal{M}_1,\ldots,\mathcal{M}_{{\tau^*}'})\nonumber\\
    &\qquad\qquad : \tau_{j_1,\ldots,j_{{\tau^*}'}} = \tau^*)\},\label{eq:decoder}
  \end{align}
  where ${\tau^*}'=\min(\tau^*,N)$ and $\max(\cdot)$ returns the tuple attaining
  the maximum in lexicographical order.
  The average transmission length is bounded as 
  \begin{align}
    \E{\tau^*} &\leq \sum_{n=1}^{\infty}\frac{1}{M_1^{n'}}
    \sum_{\stackrel{(j_1,\ldots,j_{n'})\in}{(\mathcal{M}_1,\ldots,\mathcal{M}_{n'})}}
    \pr{\tau_{j_1,\ldots,j_{n'}}=n|J_1=j_1,\ldots,J_{n'}=j_{n'}}\\
    &= \sum_{n=1}^{\infty} \pr{\tau_{1,\ldots,1}=n|J_1=1,\ldots,J_{n'}=1}\\
    &=\E{\tau}.\label{eq:proof_thm2}
  \end{align}
  where \eqref{eq:proof_thm2} follows from the definition of $\tau$ in
  \eqref{eq:tau}.  Finally, the probability of error is bounded as
  following
  \begin{align}
    &    \pr{\mbf{g}^{SF}_{\tau^*}(\mbf{Y}^{\tau^*})\not= (J_1,\ldots,J_{{\tau^*}'}) }\\
    & \quad=\EE{\tau^*}{ \pr{\mbf{g}^{SF}_{\tau^*}(\mbf{Y}^{\tau^*})\not=
        (J_1,\ldots,J_{{\tau^*}'})| \tau^* }}\\
    & \quad\leq\EE{\tau^*}{\pr{\mbf{g}^{SF}_{\tau^*}(\mbf{Y}^{\tau^*}) \not=
        (\underbrace{1,\ldots,1}_{ {\tau^{*}}' \text{ times}}) |
        J_1=1,\ldots,J_{{\tau^*}'}=1, \tau^*}}\label{eq:step1}\\
    &\quad \leq \EE{\tau^*}{\pr{ \tau_{\underbrace{1,\ldots,1}_{{\tau^*}'\text{ times}}} \geq \tau^*|J_1=1,\ldots,J_{{\tau^*}'}=1  ,\tau^*}}\label{eq:step2}\\
    &\quad\leq \E{
      \sum_{\stackrel{j_1\in\mathcal{M}_1,\ldots,j_{{\tau^*}'}\in\mathcal{M}_{{\tau^*}'}}{(j_1,\ldots,j_{{\tau^*}'})\not=(1,\ldots,1)
        }} \mathbbm{1}\{
      \tau_{j_1,\ldots,j_{{\tau^*}'}}\leq \tau_{1,\ldots,1}  \} \right.\nonumber\\
&\qquad\qquad\qquad\qquad\qquad\qquad \left. |J_1=1,\ldots,J_{{\tau^*}'}=1} \label{eq:step3} \\
    &\quad= \E{ \sum_{n=1}^{{\tau^*}'}
      \sum_{\stackrel{\stackrel{(j_1,\ldots,j_{n-1})=(1,\ldots,1),}{j_{n}\in\mathcal{M}_n,\ldots,j_{{\tau^*}'}\in\mathcal{M}_{{\tau^*}'}},}{
          j_n\not=1}} \mathbbm{1}\{
      \overline{\tau}_n\leq \tau\}}\label{eq:step4}\\
    &\quad= \E{ \sum_{n=1}^{{\tau^*}'}
      (M_{n}-1) M_{n+1}^{{\tau^*}'} \mathbbm{1}\{
      \overline \tau_n \leq \tau \}},
  \end{align}
  where \eqref{eq:step1} follows from \eqref{eq:decoder}, \eqref{eq:step3} from
  the union bound and \eqref{eq:step4} from the definition of $\overline \tau_n$ in
  \eqref{eq:overlinetau}. Lastly, \eqref{eq:vlf_eps} follows from the fact that
  ${\tau^*}'\leq \min(\tau,N)$, which completes the proof.
\end{proof}

\end{document}